\documentclass[runningheads]{llncs}
\AtBeginDocument{%
  \providecommand\BibTeX{{%
    \normalfont B\kern-0.5em{\scshape i\kern-0.25em b}\kern-0.8em\TeX}}}
\usepackage[colorlinks=true,citecolor=blue,linkcolor=blue]{hyperref}

\usepackage{listings}
\usepackage[shortlabels]{enumitem}

\lstdefinelanguage{register}{
    morekeywords={inc, dec, goto, halt, accept, reject},
    morecomment=[l]{//}, % l is for line comment
    morecomment=[s]{/*}{*/}, % s is for start and end delimiter
    morestring=[b]" % defines that strings are enclosed in double quotes
} % 
\usepackage{xcolor}
\definecolor{codegreen}{rgb}{0,0.6,0}
\definecolor{codegray}{rgb}{0.5,0.5,0.5}
\definecolor{codepurple}{rgb}{0.58,0,0.82}
\definecolor{backcolour}{rgb}{0.95,0.95,0.92}
\renewcommand{\vec}[1]{\mathbf{#1}}

\usepackage{physics}

\newcommand{\R}{\mathbb{R}}
\newcommand{\N}{\mathbb{N}}

\newcommand{\calC}{\mathcal{C}}

\newcommand{\ignore}[1]{}

\usepackage[normalem]{ulem}
\usepackage{cancel}

\newcommand{\vc}{\vec{c}}
\newcommand{\vd}{\vec{d}}

\newcommand{\vr}{\vec{r}}

\newcommand{\vs}{\vec{s}}
\newcommand{\vo}{\vec{o}}
\newcommand{\vM}{\vec{M}}

\newcommand{\vp}{\vec{p}}
\newcommand{\vi}{\vec{i}}
\newcommand{\vx}{\vec{x}}
\newcommand{\vy}{\vec{y}}

\newcommand{\vz}{\vec{z}}

\newcommand{\vh}{\vec{h}}

\renewcommand{\vu}{\vec{u}}

\renewcommand{\vb}{\vec{b}}
\newcommand{\bfr}{\mathbf{r}}
\newcommand{\bfp}{\mathbf{p}}

\newcommand{\Rp}{\mathbb{R}_{\geq 0}}

\def\longrightharpoonup{\relbar\joinrel\rightharpoonup}
\def\longleftharpoondown{\leftharpoondown\joinrel\relbar}
\def\longrightleftharpoons{\mathop{\vcenter{\hbox{\ooalign{\raise1pt\hbox{$\longrightharpoonup\joinrel$}\crcr\lower1pt\hbox{$\longleftharpoondown\joinrel$}}}}}}
\def\rxn{\mathop{\rightarrow}\limits}  % use as A+B \rxn^k C

\usepackage{mathtools}

\newcommand\rxni[1]{\xrightarrow[]{\substack{#1\\ \bot}}}

\newcommand{\slto}{\to^1}      % straight-line-reachable
\newcommand{\segto}{\rightsquigarrow}      % segment-reachable

\usepackage[textsize=tiny]{todonotes}
\setuptodonotes{fancyline}

\usepackage{amsmath}
\usepackage{amssymb}
\usepackage{wasysym}   %for \leadsto squiggle arrow
\usepackage{ifpdf}
\usepackage{comment}
\usepackage{import}
\usepackage{caption}
\usepackage{subcaption}
\usepackage{wrapfig}
\usepackage{color}
\usepackage{url}
\usepackage{amsfonts}

\usepackage[capitalize]{cleveref}
\crefname{claim}{claim}{claims}

\newtheorem{thm}{Theorem}
\newtheorem{defn}[thm]{Definition}{\bfseries}{\itshape}
{\bfseries}{\itshape}
\newtheorem{lem}[thm]{Lemma}{\bfseries}{\itshape}
{\bfseries}{\itshape}
\newtheorem{obs}[thm]{Observation}{\bfseries}{\itshape}

\numberwithin{equation}{section}
\numberwithin{thm}{section}
\crefname{thm}{Theorem}{Theorems}
\crefname{lem}{Lemma}{Lemmas}
\crefname{obs}{Observation}{Observations}
\crefname{defn}{Definition}{Definitions}
\crefname{defn}{Definition}{Definitions}

\usepackage{multirow}

\begin{document}

\title{Rate-independent continuous inhibitory chemical reaction networks are Turing-universal}
\titlerunning{Rate-independent continuous inhibitory CRNs are Turing-universal}
\author{Kim Calabrese\inst{1} \and
David Doty\inst{1}}
\institute{University of California, Davis, USA\\
\email{\{ebcalabrese,doty\}@ucdavis.edu}}

% \author{Kim Calabrese}
% \email{ebcalabrese@ucdavis.edu}
% \affiliation{%
%   \institution{University of California, Davis}
%   \streetaddress{One Shields Ave}
%   \city{Davis}
%   \country{USA}
% }

% \author{David Doty}
% \email{doty@ucdavis.edu}
% \affiliation{%
%   \institution{University of California, Davis}
%   \streetaddress{One Shields Ave}
%   \city{Davis}
%   \country{USA}
% }

\maketitle

% !TEX root = main.tex

\begin{abstract}
We study the model of continuous chemical reaction networks (CRNs), consisting of reactions such as $A+B \rxn C+D$ that can transform some continuous, nonnegative real-valued quantity (called a \emph{concentration}) of chemical species $A$ and $B$ into equal concentrations of $C$ and $D$.
Such a reaction can occur from any state in which both reactants $A$ and $B$ are present, i.e., have positive concentration.
We modify the model to allow \emph{inhibitors}, for instance, reaction $A+B \rxni{I} C+D$ can occur only if the reactants $A$ and $B$ are present and the inhibitor $I$ is absent.

The computational power of non-inhibitory CRNs has been studied.
For instance, the reaction $X_1+X_2 \rxn Y$ can be thought to compute the function $f(x_1,x_2) = \min(x_1,x_2)$.
Under an ``adversarial'' model in which reaction rates can vary arbitrarily over time, it was found that exactly the continuous, piecewise linear functions can be computed,
ruling out even simple functions such as $f(x) = x^2$.
In contrast, in this paper we show that inhibitory CRNs can compute any computable function $f:\N\to\N$.
\end{abstract}


\section{Introduction}
\label{sec:intro}

The model of continuous chemical reaction networks (CRNs) consists of reactions such as $A+B \rxn C+D$ that can transform some continuous, nonnegative real-valued quantity (called a \emph{concentration}) of chemical species $A$ and $B$ (the \emph{reactants}) into equal concentrations of $C$ and $D$ (the \emph{products}).
This model has long held an important role in modeling naturally occurring chemical systems and predicting their evolution over time.
Recently, the model has been investigated, not as a modeling language, but as a \emph{programming} language for describing desired behavior of engineered chemicals.
For example, the reaction $X_1 + X_2 \rxn Y$ can be thought to compute the function $f(x_1,x_2) = \min(x_1,x_2)$,
in the sense that if we start in configuration $\{x_1 X_1, x_2 X_2\}$, i.e., concentration $x_1$ of species $X_1$ and concentration $x_2$ of species $X_2$, as long as the reaction keeps happening, it will eventually produce concentration $\min(x_1,x_2)$ of species $Y$. 

The computational power depends greatly on how reaction rates are defined. 
The most common rate model is \emph{mass-action},
which says that the rate of a reaction like $A+B \rxn^{k} C+D$, with positive \emph{rate constant} $k > 0$, proceeds at rate $k \cdot [A] \cdot [B]$, where $[S]$ represents the concentration of species $S$.
The rates of all reactions affecting a species $S$ determines its derivative $\frac{d[S]}{dt}$ (adding rates of reactions where $S$ is a product, and subtracting rates where it is a reactant),
so the concentrations evolve according to a system of polynomial ODEs.
It was recently shown that mass-action CRNs are capable of Turing universal computation~\cite{fages2017strong},
a very complex construction resulting from a long and deep line of research that culminated in showing the surprising computational power of polynomial ODEs~\cite{bournez2017odes}.

What if reaction rates are not so predictable over time?
One could imagine a solution does not remain well-mixed, so that some reactions go faster in a certain part of the volume where some species are more concentrated.
It is also the case that it is difficult experimentally to engineer precise rate constants~\cite{srinivas2017enzyme}.
To address these issues, 
Chen, Doty, Reeves, and Soloveichik~\cite{chen2023rate} defined a model of \emph{adversarial} reaction rates and asked what functions can be computed when the rates can vary arbitrarily over time.
They found that this model, called \emph{stable computation}, is much more computationally limited than with mass-action rates:
exactly the continuous, piecewise linear functions $f:\R^d \to \R$ can be stably computed.\footnote{
    Technically this is using the so-called \emph{dual-rail} encoding, which represents a single real value $x$ as the difference of \emph{two} species concentrations $[X^+] - [X^-]$.
    If one encodes inputs and output directly as nonnegative concentrations,
    then some discontinuities can occur, but only when some input $x_i$ goes from 0 to positive.
}
An open question from~\cite{chen2023rate} concerns a natural modification of the CRN model,
inspired by similar models of gene regulatory networks,
in which the presence of a species can \emph{inhibit} a reaction from occurring.
For example,
the reaction $A+B \rxni{I} C+D$ can occur only if its reactants are present ($[A],[B] > 0$) and its inhibitor is absent ($[I]=0$).
We call such a network an \emph{inhibitory chemical reaction network} (iCRN).\footnote{Note that our notation $A+B \rxni{I} C+D$ puts inhibitors above the reaction arrow where a rate constant would normally be written, but since we consider rate-independent computation, we will have no rate constants. We also note that in gene regulatory networks, typically a species (called \emph{transcription factor} in that literature) inhibits another \emph{species}, which is assumed to be produced at some otherwise constant rate by a single reaction, whereas our model is more general in allowing inhibitors of arbitrary reactions (so $I$ could inhibit production of $C$ via one reaction $A \rxni{I} C$ but not via another reaction $B \rxn C$.)}

The negative results of~\cite{chen2023rate}, 
showing computation is limited to continuous piecewise linear functions, 
heavily use the fact that the reachability relation $\segto$ (defined in \Cref{sec:prelim}) on CRN configurations is \emph{additive}:
if $\vx \segto \vy$ for configurations $\vx,\vy$ (nonnegative vectors representing concentrations of each species), then for all nonnegative $\vc$, we have $\vx+\vc \segto \vy+\vc$;
in other words the presence of extra molecules (represented by $\vc$) cannot \emph{prevent} reactions from occurring.
However, with inhibitors reachability is no longer additive (if $\vc$ contains inhibitors that are absent in $\vx$),
so it is natural to wonder if inhibitors increase the computational power of the model.

It is well-known ``folklore'' that in the \emph{discrete} model of iCRNs,
where the amount of a species is modeled as a nonnegative integer \emph{count},
in which reactions discretely increment or decrement species counts,
then inhibitors give the model Turing-universal power.
It is worth seeing why this is true, to understand the novel contribution of this paper (and why it is not trivially solved in the continuous model by the discrete iCRN we describe next).
It is well-known that register machines---finite-state machines equipped with a fixed number of nonnegative integer \emph{registers}, each of which can be incremented, decremented, or tested for 0---are Turing universal~\cite{minsky1967computation}.
An example register machine is:
\begin{lstlisting}
dec r1,5
inc r2
inc r2
goto 1
halt
\end{lstlisting}
Line (a.k.a., \emph{state}) 1 has the interpretation:
decrement register $r_1$ and then go to line 2, unless $r_1$ is 0, in which case go to line 5.
Increment instructions always increment the specified register and go to the next line.
The \lstinline{goto 1} statement on line 4 is syntactic sugar for \lstinline{dec r3,1} for some register $r_3$ that is always 0.
The above register machine, interpreted as taking an input $x$ in register $r_1$ and halting with an output value in register $r_2$, computes the function $f(x) = 2x$.

For a register machine consisting of such increment and decrement instructions, 
the following is a straightforward transformation of the instruction for line/state $i$ to iCRN reactions:
\begin{center}
\begin{tabular}{|l|l|}
  \hline
  % \verb!goto k! & $L_i \rxn L_{k}$
  % \\ \hline
  \verb!inc r_j! & $L_i \rxn L_{i+1} + R_j$
  \\ \hline
  \verb!dec r_j,k! & $L_i + R_j \rxn L_{i+1}$
  \\
                   & $L_i \rxni{R_j} L_k$
                     %\end{eqnarray*}
  \\ \hline
\end{tabular}
\end{center}

It is clear that at any time exactly one reaction is applicable, and it simulates the next instruction of the register machine.
In particular, when on a decrement instruction, the power of inhibition is used to ensure that if $R_j$ has positive count, then only the first of the two decrement reactions is applicable
(and as in the non-inhibitory CRN model, when $R_j$ is absent, only the second decrement reaction is applicable).
Note that a \lstinline{halt} instruction on line $i$ is not explicitly implemented as any reaction;
the simple lack of any reaction with $L_i$ as a reactant means that the CRN will terminate when the register machine does.

Our main construction in \Cref{sec:results} follows this basic strategy of simulating register machines, using inhibition to detect when a register is 0.
However, our main novel contribution is a way to ``discretize'' the behavior of the continuous CRN, so that the discrete steps of the register machine can be simulated faithfully.
This is primarily done by introducing a \emph{stable oscillator}, shown in \Cref{sec:stable-oscillation}.
\section{Preliminaries}
\label{sec:prelim}
These definitions largely follow those of~\cite{chen2023rate},
the only exception being the definition of \emph{applicable reaction}, which is modified to account for inhibitors.

For any set $A$, let $\mathcal{P}(A)$ denote the power set of $A$ (set of all subsets of $A$).
Let $\N$ denote the nonnegative integers and $\R$ denote the real numbers.
Given a finite set $F$ and a set $S$,
let $S^F$ denote the set of functions $\vc: F \to S$.
In the case of $S = \R$ (resp. $\N$), we view $\vc$ equivalently as a real-valued (resp. integer-valued) vector indexed by elements of $F$.
Given $x \in F$, we write $\vc(x)$, to denote the real number indexed by $x$.
The notation $\Rp^F$ is defined similarly for nonnegative real vectors.
Throughout this paper, let $\Lambda$ be a finite set of chemical \emph{species}.
Given $S\in \Lambda$ and $\vc \in \Rp^\Lambda$, we refer to $\vc(S)$ as the \emph{concentration of $S$ in $\vc$}.
When the configuration $\vc$ is understood from context, we write $[S]$ to denote $\vc(S)$.
For any $\vc\in \Rp^\Lambda$, let $[\vc] = \{S \in \Lambda \ |\ \vc(S) > 0 \}$, the set of species \emph{present} in $\vc$
(a.k.a., the \emph{support} of $\vc$).
We write $\vc \leq \vc'$ to denote that $\vc(S) \leq \vc'(S)$ for all $S \in \Lambda$.
Given $\vc,\vc' \in \Rp^\Lambda$, we define the vector component-wise operations of addition $\vc+\vc'$, subtraction $\vc-\vc'$, and scalar multiplication $x \vc$ for $x \in \R$.
% If $\Delta \subset \Lambda$, we view a vector $\vc \in \Rp^\Delta$ equivalently as a vector $\vc \in \Rp^\Lambda$ by assuming $\vc(S)=0$ for all $S \in \Lambda \setminus \Delta.$
% For $\Delta \subset \Lambda$, we write $\vc  \upharpoonright \Delta$ to denote $\vc$ \emph{restricted to} $\Delta$; in particular, $\vc  \upharpoonright \Delta = \vec{0} \iff (\forall S\in\Delta)\ \vc(S)=0.$
% (We use the convention that $\vc \upharpoonright \emptyset = \vec{0}$ for all configurations $\vc$.)

A \emph{reaction} over $\Lambda$ is a triple $\alpha = (\bfr,\Delta,\bfp) \in \N^\Lambda \times \mathcal{P}(\Lambda) \times \N^\Lambda$, 
such that $\vr \neq \vp$,
specifying the stoichiometry of the reactants, products, as well as the inhibitors of the reaction respectively.\footnote{It is customary to define, for each reaction, a \emph{rate constant} $k \in \R_{>0}$ specifying a constant multiplier on the mass-action rate (i.e., the product of the reactant concentrations), but as we are studying CRNs whose output is independent of the reaction rates, we leave the rate constants out of the definition.} %, and the \emph{rate constant} $k$.
We say a reaction $\alpha$ is \emph{inhibited} by species $I$ if $I \in \Delta$. For instance, given $\Lambda=\{A,B,C,I\}$, the reaction $A+2B \rxni{I} A+3C$ is the triple $({(1,2,0)},{\{I\}},{(1,0,3)}).$ 
% We represent reversible reactions such as $A \revrxn B$ as two irreversible reactions $A \to B$ and $B \to A$.
% In this paper, we assume that $\bfr \neq \vec{0}$, i.e., we have no reactions of the form $\emptyset \to \ldots$.
 
% \footnote{We allow high order reactions; i.e., those that have more than two reactants.
% Such higher order reactions could be eliminated from our constructions using the transformation that replaces $S_1 + S_2 + \ldots + S_n \to P_1 + \ldots + P_m$ with bimolecular reactions $S_1 + S_2 \revrxn S_{12}, S_{12} + S_3 \revrxn S_{123}, S_{123} + S_4 \revrxn S_{1234}, \ldots, S_n + S_{12 \ldots n-1} \to P_1 + \ldots + P_m$.
% }

An \emph{inhibitory chemical reaction network (iCRN)} is a pair $\calC=(\Lambda,R)$, where $\Lambda$ is a finite set of chemical \emph{species},
and $R$ is a finite set of reactions over $\Lambda$.
A \emph{configuration} of a iCRN $\calC=(\Lambda,R)$ is a vector $\vc \in \Rp^\Lambda$.
Given a configuration $\vc$ and reaction $\alpha = (\bfr,\Delta,\bfp)$, we say that $\alpha$ is \emph{applicable} in $\vc$ if $[\bfr] \subseteq [\vc]$ (i.e., $\vc$ contains positive concentration of all of the reactants) and $[\vc] \cap \Delta = \emptyset$ (no inhibitor is present in $\vc$).
If no reaction is applicable in configuration $\vc$, we say $\vc$ is \emph{static}.
% We say a species $S$ is \emph{produced} in reaction $\langle  \vr,\vp,\Delta\rangle$ if $\vr(S) < \vp(S)$,
% and
% \emph{consumed} if $\vr(S) > \vp(S)$.
% (Note that a catalyst, such as $C$ in the reaction $C+X \to C+Y$, is neither produced nor consumed.)

Fix an iCRN $\calC = (\Lambda,R)$. We define the 
$|\Lambda| \times |R|$ \emph{stoichiometry matrix} $\vM$ such that,
for species $S \in \Lambda$ and reaction $\alpha = (\vr,\Delta,\vp) \in R$,
$\vM(S,\alpha) = \vp(S) - \vr(S)$
is the net amount of 
$S$ produced by $\alpha$ (negative if $S$ is consumed).\footnote{$\vM$ does not fully specify $\calC$, since catalysts and inhibitors are not modeled: reactions $A + B \rxni{C} A + D$ and $B \rxn D$ both correspond to the column vector $(0,-1,0,1)^\top$.}
For example, if we have the reactions $X \to Y$ and $X + A \to 2X + 3Y$, and if the three rows correspond to $A$, $X$, and $Y$, in that order, then
  $$
    \vM =
    \left(
      \begin{array}{cc}
         0 & -1 \\
        -1 &  1 \\
         1 &  3 \\
      \end{array}
    \right)
  $$

\begin{defn}\label{defn-reachable-line}
Configuration $\vd$ is \emph{straight-line reachable (aka $1$-segment reachable)} from configuration $\vc$, written $\vc \slto \vd$, if $(\exists \vu \in \Rp^R)\ \vc + \vM \vu = \vd$
and $\vu(\alpha) > 0$ only if reaction $\alpha$ is applicable at $\vc$.
In this case write $\vc \slto_\vu \vd$.
\end{defn}

\noindent Intuitively, by a single segment we mean running the reactions applicable at $\vc$ at a constant (possibly 0) rate to get from $\vc$ to $\vd$.
In the definition, $\vu(\alpha)$ represents the flux of reaction $\alpha \in R$.

\begin{defn}
\label{defn:reachable-lines}
Let $k \in \N$.
Configuration $\vd$ is \emph{$k$-segment reachable} from configuration $\vc$, written $\vc \segto^k \vd$, 
if $(\exists \vb_0, \dots, \vb_{k})\ \vc  = \vb_0 \slto \vb_1 \slto \vb_2 \slto \dots  \slto \vb_{k}$,
with $\vb_k = \vd$.
\end{defn}

\begin{defn}
\label{defn:reachable-segment}
Configuration $\vd$ is \emph{segment-reachable} 
(or simply \emph{reachable}) 
from configuration $\vc$, written $\vc \segto \vd$, if $(\exists k\in\N)\ \vc \segto^k \vd$.
\end{defn}

\noindent Often Definition \ref{defn:reachable-segment} is used implicitly, when we make statements such as, ``Run reaction 1 until $X$ is gone, then run reaction 2 until $Y$ is gone'', which implicitly defines two straight lines in concentration space.
Although we make no attempt to ascribe an ``execution time'' to any path followed by segments in Definition~\ref{defn:reachable-segment}, it is sometimes useful to refer to such paths over time.
In this case we suppose that each segment takes one unit of time,
so that if $\vx \segto^k \vy$,
we associate this to a \emph{trajectory} $\rho:[0,k] \to \Rp^\Lambda$, where $\rho(t)$ represents the concentrations of species after $t$ units of time have elapsed,
i.e., following the first $\lfloor t \rfloor$ segments, then a fraction of the $t$'th segment if $t \not \in \N$ (so that for integer $t$,
$\rho(t)$ is the configuration $\vb_t$ in Definition~\ref{defn:reachable-lines}).
In this case we write $\vx \segto_\rho \vy$.

Given configurations $\vx,\vy,\vz$ such that $\vx \segto_{\rho_1} \vy$ and $\vy \segto_{\rho_2} \vz$,
we denote the \emph{concatenation} of trajectories $\rho_1$ and $\rho_2$  
to be the trajectory $\rho_1 : \rho_2$ such that $\vx \segto_{\rho_1:\rho_2} \vz$. 

We now formalize what it means for an iCRN to ``rate-independently'' compute a function $f$.
Since our main result is about simulating register machines that process natural numbers,
we define stable computation for functions $f: \N \to \N$.\footnote{
    Since iCRNs operate on real-valued concentrations,
    a very similar definition for functions $f: \Rp \to \Rp$ makes sense (and was formally defined for non-inhibitory CRNs in~\cite{chen2023rate});
    \Cref{sec:conclusion} discusses this issue further.
    We could also extend the definition to take multiple inputs for a function $f:\N^d \to \N$, but since register machines are Turing universal,
    we could encode multiple input integers via a pairing function into a single integer,
    so it is no loss of generality to consider single-input functions.
}
An \emph{inhibitory chemical reaction computer (iCRC)} is a tuple $\mathcal{C} = (\Lambda, R, \vs, X, Y)$,
where $(\Lambda,R)$ is an iCRN,
$\vs \in \N^\Lambda$ is the \emph{initial context}
(species other than the input that are initially present with some constant concentration; in our case, $\vs(A_1) = 1$ for a single species $A_1 \in \Lambda$ and $0$ for all other species),
$X \in \Lambda$ is the \emph{input species},
and $Y \in \Lambda$ is the \emph{output species}.
We say a configuration $\vo \in \Rp^\Lambda$ is \emph{stable} if, for all $\vo'$ such that $\vo \segto \vo'$, $\vo(Y) = \vo'(Y)$,
i.e., the concentration of $Y$ cannot change once $\vo$ has been reached.
Let $f: \N \to \N$.
We say $\mathcal{C}$ \emph{stably computes} $f$ if,
for all $n \in \N$,
starting from initial configuration $\vi = \vs + \{n X\}$
(i.e., starting with initial context, plus the desired input amount of $X$),
for all configurations $\vc$ such that $\vi \segto \vc$, 
there is $\vo$ such that $\vc \segto \vo$,
such that $\vo$ is stable and $\vo(Y) = f(n)$.

\section{Main results}
\label{sec:results}

Our goal is to design an iCRN that simulates the behavior of a register machine, 
similar to simulations by discrete CRNs~\cite{SolCooWinBru08,angluin2006fast}.
The inclusion of inhibitors to our model allows us to enforce deterministic state transitions in chemical reaction networks, but to emulate the sequential power of discrete computation, we need a mechanism to manage control flow. 
First, we describe a simpler ``stably oscillating'' iCRN that is, in a sense, the main conceptual contribution of this paper.

\subsection{Stable oscillation}
\label{sec:stable-oscillation}

% It may be necessary to do case analysis based on the number of species in our CRN. If we have an even number of species, then we need to oscillate every other every other even species to avoid violating  If we have 3 species $A$, $B$ and $C$ and we want to oscillate Not really we can just oscillate the odd ones and that works i think

The following definition captures the behavior of a system of chemical reactions that execute sequentially, and eventually repeat their execution.
A similar definition for the discrete model of population protocols appears in \cite{cooper2017constructing}.\footnote{
    However, Definition~\ref{defn:stable-oscillation} is distinct from the that of~\cite{cooper2017constructing},
    both by being defined in a continuous-state rather than a discrete-state model,
    and in that we do not require ``self-stabilizing'' behavior (which dictates that the behavior should occur from any possible initial state).
}
In particular, we have species $A_1,\dots,A_k$ that all start at 0.
$A_1$ monotonically goes up to 1, then monotonically down to 0, 
then $A_2$ goes up and down similarly,
etc.
After $A_k$ does this,
the whole thing repeats.

\begin{defn}
\label{defn:period-of-oscillation}
Let $\mathcal{A} = \{A_1,A_2,\ldots,A_k\}$ be a set of species in an iCRN, and let $\rho$ be a trajectory. We say $\rho([t_1,t_2])$ is a \emph{wave of $A_i$} if for some $t_1 < t < t_2$

\begin{itemize}
    \item $\rho(t_1)(A_i) = \rho(t_2)(A_i) = 0$,
    \item $\rho(t)(A_i) = 1$,
    \item $\rho([t_1,t])(A_i)$ is nondecreasing, and
    \item $\rho([t,t_2])(A_i)$ is nonincreasing.
\end{itemize}
$\rho([T_1,T_2])$ is a \emph{period of oscillation} of $\mathcal{A}$ if there exists $T_1 = t_1, t_2, \ldots, t_k = T_2$ such that for all $0 \le i < k$,

\begin{itemize}
    \item $\rho([t_i,t_{i+1}])$ is a wave of $A_i$, and
    \item for all $j \ne i$ and all $t_i \leq t \leq t_{i+1}$, $\rho(t)(A_j) = 0$.
\end{itemize}
% We say it \emph{fairly oscillates} on $\mathcal{A}$ if for all fair executions $f$ and $p \in \N$ there's some $t > 0$ such that $f([0,t])$ is $p$ periods of oscillation.
\end{defn}

\begin{defn}
\label{defn:stable-oscillation}
\noindent We say an iCRN $\mathcal{C}$ \emph{stably oscillates} on $\mathcal{A}$ from configuration $\Vec{i}$ if for all $\Vec{c}$ such that $\Vec{i} \rightsquigarrow_{\rho_1} \Vec{c}$, we have $\Vec{c} \rightsquigarrow_{\rho_2} \Vec{i}$ such that letting $\rho = \rho_1:\rho_2$, $\rho([0,t])$ is one or more periods of oscillation of $\mathcal{A}$.
\end{defn}

The next lemma demonstrates an iCRN that stably oscillates.
We note that Lemma~\ref{lem:oscillator} is not used directly in the rest of the paper.
Instead, the proof of Lemma~\ref{lem:oscillator} is intended to serve as a ``warmup'' to illustrate some of the key ideas used in the more complex iCRN defined in \Cref{sec:reg-machine-simulation-construction}.

\begin{lem}\label{lem:oscillator}
Let $n \geq 3$ and $\mathcal{C}$ be the iCRN with species $\Lambda = \{X_0,X_1,\ldots,X_{n-1}\}$ and for each $0 \leq i < n$, reaction 
$
    X_i \rxni{X_{i-1}} X_{i+1},
$
where $i - 1$ and $i + 1$ are both taken modulo $n$. If $\vi = \{1X_0\}$ is the starting configuration, then $\mathcal{C}$ stably oscillates on $\mathcal{O} = \{X_i \mid 0 \le i \le n \mathrm{~, } ~ i ~\mathrm{ is ~ odd} \}$. 
\end{lem}

\begin{proof}
For each $0 \leq i < n$, let $\alpha_i$ be the reaction $X_i \rxni{X_{i-1}} X_{i+1}$. First, observe that for any configuration $\vc$ in which the species $X_i$ and $X_{i+1}$ are present, the only applicable reaction is $\alpha_i$,
since the reaction $X_{i+1} \rxni{X_{i}} X_{i+2}$ is inhibited by $X_i$,
and all other reactions have a reactant absent. 
Thus every sufficiently long path from $\vc$ just executes $\alpha_i$ until $X_i$ is absent.
Once $X_i$ is absent, $\alpha_{i+1}$ becomes applicable.
At this point, we have only $X_{i+1}$ present, so by similar reasoning, only $\alpha_{i+1}$ is applicable and every sufficiently long path runs only $\alpha_{i+1}$ until $X_{i+1}$ is absent.

Iterating this reasoning over all $i$,
for each $0 \leq i < n$, let $\vu_i$ denote the flux vector with $\vu_i(\alpha_i) = 1$ and $\vu_i(\alpha_j) = 0$ for $j \neq i$ (i.e., execute only reaction $\alpha_i$, for one unit of flux).
Then starting from initial configuration $\vi = \{1 X_0\}$, we see that every path starting from $\vi$ is of the form 
\begin{align*} \hspace*{-1.5cm} 
  &\{1 X_0\}
  \slto_{\vu_0}
  \{1 X_1\}
  \slto_{\vu_1}
  \{1 X_2\}
  \slto_{\vu_2}
  \ldots
  \\
  &\{1 X_{n-1}\}
  \slto_{\vu_{n-1}}
  \{1 X_0\}
  \slto_{\vu_0}
  \{1 X_1\}
  \slto_{\vu_1}
  \ldots
  \{ a X_i, (1-a) X_{i+1}\},
\end{align*}
for some $0 \leq a \leq 1,$
or, assuming the path does not get to configuration $\{1 X_0\}$ above,
$
  \{1 X_0\}
  \slto_{\vu_0}
  \{1 X_1\}
  \slto_{\vu_1}
  \{1 X_2\}
  \slto_{\vu_2}
  \ldots
  \{ a X_i, (1-a) X_{i+1} \}.
$

In either case, by continuing to apply $\alpha_i$ with flux $a$, then unit fluxes of $\alpha_{i+1}, \alpha_{i+2}$, etc. until we reach configuration $\{1 X_0\}$,
this does some positive integer number of periods of oscillation.
Let $\vi = \{1 X_0\}$, $\vc = \{ a X_i, (1-a) X_{i+1} \}$,
this satisfies the definition of oscillation for the species in $\mathcal{O}$.
\qed\end{proof}

\subsection{Construction of iCRN simulating a register machine}
\label{sec:reg-machine-simulation-construction}
In this section we describe how to construct an iCRN $\mathcal{C}$ to simulate an arbitrary register machine $\mathcal{R}$.

Let the set of  states (or lines) of $\mathcal{R}$ be $Q = \{1,2,\dots,m\},$
supposing it starts in state $1$ with initial input register value $n \in \N$.
Suppose $\mathcal{R}$'s input register is \lstinline{r_in} and its output register is \lstinline{r_out}.
To simulate $\mathcal{R}$, $\mathcal{C}$ has input species $R_\text{in}$ and output species $R_\text{out}$, and starts with configuration $\{1 A_1, n R_\text{in}\}$
(i.e., with initial context $\vs = \{1 A_1\}$).

Consider these reactions, which implement the stable oscillator of Lemma~\ref{lem:oscillator} with 3 species (where $X_0=A, X_1=B, X_2=C$).
\begin{align*}
    A &\rxni{C} B
    \\
    B &\rxni{A} C
    \\
    C &\rxni{B} A
\end{align*}
Although we do not use those exact reactions,
it is helpful to see that iCRN as an introduction to how we implement the oscillator component of $\mathcal{C}$.
$\mathcal{C}$ has $m$ variants of each of those species $\{A_1,B_1,C_1,\dots,A_m,B_m,C_m\}$, each subscript representing a state of $\mathcal{R}.$
We will additionally have species $R_1,R_2,\dots$ to represent the various registers of $\mathcal{R}$ as well as designated input and output species $R_\text{in},R_\text{out}$.
For ease of exposition, we use the convention that $\mathcal{R}$ has exactly one input and output register, but this is easily extendable.

Intuitively, the variants of the last reaction $C \rxni{B} A$ will perform all the logic of the register machine: incrementing, decrementing, and changing states.
The other two (variants of) reactions $A \rxni{C} B$ and $B \rxni{A} C$ are simply to make the oscillator work while remembering the current state.
However, since the stateful oscillator will change states in the last reaction, and the last reaction's reactant is an inhibitor for the first reaction, we need to be careful in selecting the correct inhibitors for the first reaction to acknowledge the states are different, and that \emph{multiple} stateful variants of $C$ could be inhibitors of a single variant of $A$.

Formally, for all $1 \leq i \leq m$, 
$\mathcal{C}$ has the reaction 
\[B_i \rxni{A_i} C_i.\]
For all $1 \leq i \leq m$, let $\{{j_1}, {j_2}, \ldots, {j_l}\}$ be the set of states that are potential predecessors of state $i$.
This includes $j = i-1$ if $i > 1$ and state $j$ is not a \lstinline{goto}, as well as all $j$ such that a decrement test for 0 can cause a jump from $j$ to $i$.
For all $1 \leq i \leq m$, 
$\mathcal{C}$ also has the reaction
\[A_i \rxni{C_{j_1},C_{j_2},\dots,C_{j_l}} B_i.\]
Finally, for all $1 \leq i \leq m$,
$\mathcal{C}$ has the following reactions to simulate register machine instructions.
% KC: I went through and changed the reactions to perform instructions on the reactions with C as a reactant like we discussed. I do think this is more intuitive, but I do sort of like how A_1 is in the starting configuration and also triggers reactions
\begin{itemize}
\item
    \label{rxn:inc}
    if state $i$ is \lstinline{inc r_j} (increment register $j$ and move from state $i$ to $i+1$):
    \begin{align*}
        C_i &\rxni{B_i} A_{i+1} + R_j 
        % \\
        % B_{i+1} &\rxni{A_i} C_{i+1}
    \end{align*}
\noindent Note the dual role of $C_i$:
it helps the ``clock'' to oscillate, but its maximum concentration also defines one ``unit'' of concentration to help us use real-valued concentrations to represent discrete integer counts in registers of $\mathcal{R}$.
In other words, the initial amount of $A_1$ (which sets the maximum concentration achieved by any $C_i$) is also the amount by which $[R_j]$ increases (and the amount it decreases in a decrement instruction).

\item
    if state $i$ is \lstinline{dec r_j,k}
    (decrement register $j$ and move from state $i$ to $i+1$, unless it is 0, in which case go to state $k$):
    \label{rxn:dec}
    \label{rxn:goto}
    \begin{align*}
        C_i + R_j &\rxni{B_i} A_{i+1}
        \\
        C_i &\rxni{B_i,R_j} A_{k}
        % \\
        % B_{i+1} &\rxni{A_i} C_{i+1}
    \end{align*}
\end{itemize}

\noindent As in the case for the discrete iCRN described in \Cref{sec:intro},
no reactions are associated to $C_i$ if state $i$ is a \lstinline{halt} instruction. 

% \todo{DD: I'm not sure what this paragraph is accomplishing. It's still somewhat ``in isolation''; I suggest we show the complete iCRN for a complete register machine after this proof}
% Together, a complete translation of a \lstinline{inc r_j} instruction would look like: 
% % \todoi{KC: I don't know how to typeset this nicely}
% \begin{eqnarray*}
%     A_i &\rxni{C_i}& B_{i+1} + R_j 
%     \\
%     B_i &\rxni{A_{j_1},A_{j_2},\dots,A_{j_l}}& C_i
%     \\
%     C_i &\rxni{B_i}& A_i
% \end{eqnarray*}

% By properties of stable oscillation, for every c reachable from the start theres a i reachable 

% Period of oscillation -> B went from 0 to 1 to 0 monotonically the only way for this to happen is for reaction A to happen with one unit of flux then reaction 3 to happen with one unit of flux.

\subsection{Proof of correctness}
In this section, we prove that the iCRN $\mathcal{C}$ described in \Cref{sec:reg-machine-simulation-construction} correctly simulates the register machine $\mathcal{R}$.

In the definition of $\segto$, it is technically allowed for two consecutive segments to ``point the same direction'', i.e., $\vx \to_{\vu_1}^1 \vy \to_{\vu_2}^1 \vz$ such that $\vu_1$ and $\vu_2$ are multiples of each other.
The next observation says that we can assume without loss of generality this does not happen, since any two such consecutive segments $\vu_1$ and $\vu_2$ can always be concatenated into a single segment $\vu_1+\vu_2$.

\begin{obs}
\label{obs:parallel-segments-concat}
    In any iCRN, if $\vx \segto \vy$, we may assume without loss of generality that each pair of consecutive segments are not multiples of each other.
    In particular, if exactly one reaction is applicable at any time, then any two consecutive segments use different reactions.
\end{obs}

\noindent We also note that there is a distinction between the function of species that ``oscillate'' (i.e. species $A_1,B_1,C_1,\ldots,A_n,B_n,C_n$) and species that represent the value stored in a register ($R_1,R_2\ldots$). We call the former \emph{oscillator species} and the latter \emph{register species}. 
Since the control flow of our construction is driven primarily by the so-called oscillator species, it suffices to focus on their behavior when discussing the properties of the iCRN induced by our construction. 

We develop machinery to talk about specific configurations of $\mathcal{C}$ that contain oscillator species at concentration 1. 
\begin{defn}
\label{defn:transition-point}
Let $A \in \Lambda$ be an oscillator species. We say configuration $\vx \in \Rp^\Lambda$ is a $\emph{transition point}$ of $A$ if $\vx(A) = 1$ and $\vx(B) = 0$ for all other oscillator species $B$.
\end{defn}

\noindent Intuitively, a transition point marks the peak of a species' oscillation, representing a configuration where a previously present oscillator species depletes, allowing a new reaction to become applicable. 
Definition~\ref{defn:transition-point} implicitly characterizes the configurations in $\mathcal{C}$: a configuration is either a transition point or lies ``between'' two transition points.
Furthermore, if a configuration is not a transition point, then the applicable reaction is exactly that applicable in the last reached transition point.

For example, the 3-species oscillator described at the start of \Cref{sec:reg-machine-simulation-construction} has $A, B$, and $C$ as oscillator species.
Consider the transition point $\{1A\}$. In this configuration, the only applicable reaction is the reaction $\alpha = A \rxni{C} B$,
since $A$ is the only species present.
Running $\alpha$ with flux $\frac{1}{2}$, we reach the configuration $\{\frac{1}{2} A, \frac{1}{2} B\}$.
Notice that even though we have some amount of $B$ present in this reaction, $\alpha$ is \emph{still} the only applicable reaction,
since $A$ inhibits the reaction $\beta = B \rxni{A} C$. 
$\beta$ only becomes applicable once we reach the configuration $\{1B\}$, but this is a transition point. This behavior can be generalized as follows:

\begin{obs}
    \label{obs:transition-pts-rxn}
    Let $\vx$, $\vy$ be configurations of the iCRN $\mathcal{C}$ described in \Cref{sec:reg-machine-simulation-construction}.
    If $\vx$ is a transition point of $A$, $\vy$ is not a transition point, 
    and
    $\vx \to^1 \vy$, then the reactions applicable in $\vy$ are exactly the reactions applicable in $\vx$.
\end{obs}

% Observe that in our construction, any transition point with an applicable reaction can always reach another transition point. 
\noindent This observation indicates that the applicable reactions of $\mathcal{C}$ changes only upon reaching a new transition point. Therefore, instead of reasoning about arbitrary configurations in concentration space, we can just consider the reachability of transition points. Additionally, observation \ref{obs:parallel-segments-concat} implies that we can assume transition points are reached in a single flux 1 line segment, enabling discrete arguments about the behavior of our construction.

% For instance, the 3-species oscillator presented in \ref{sec:configurationful-stable-oscillation} has $A, B$, and $C$ as oscillator species and in the configuration $\{1A\}$, we can reach the configuration $\{1B\}$ by performing the reaction $\alpha = A \rxniC B$ with flux $\frac{1}{2}$ twice. We can concatenate these segments together to get a segment performing $\alpha$ once. 

% By considering these particular configurations, we can make discrete arguments about the behavior of our construction even within a continuous setting.

% \todoi{explajn reachable configs are rither a transition or between two transition points and can reach one}

% Notice that from any non-terminal configuration $\vc$, a transition point $\vt$ is always reachable since an oscillator species is present. Furthermore, since at most one reaction $\alpha$ is applicable in $\vc$ and any configuration reachable from $\vc$, a sufficiently long execution of $\alpha$ will eventually reach $\vt$.

\begin{thm}
    Suppose that $\mathcal{R}$ computes a function $f:\N \to \N$ in the sense that, starting with input register having value $n$,
    it halts with output register having value $f(n)$.
    Then the iCRN $\mathcal{C}$ described above stably computes $f$ from the initial configuration $\vi = \{1A_1,nR_{\emph{in}}\}$.
\end{thm}

\begin{proof}
    A complete example of this construction is given in \Cref{sec:example}.

    Let $R_{\text{in}}$ be the input species and $R_{\text{out}}$ be the output species. 
    For $\mathcal{C}$ to stably compute $f$,
    we need that for any valid initial configuration $\vi = \{1 A_1, n R_\text{in}\}$,
    and any configuration $\vc$ such that $\vi \segto \vc$, there exists a configuration $\vo$ such that $\vc \segto \vo$, $\vo(R_{\text{out}}) = f(n)$ and for all $\vo'$ such that $\vo \segto \vo'$, $\vo'(R_{\text{out}}) = \vo(R_{\text{out}})$. 
    
    It suffices to show that for any integer initial concentration of $R_{\text{in}}$, there exists exactly one trajectory, ending in a static configuration $\vh$ such that $\vh(R_{\text{out}}) = f(n)$.

    We first prove that the following invariants hold at every reachable transition point $\vx$.
    \begin{enumerate}[(a)]
        \item 
        \label{inv-int}
        For every register species $R_j$, $\vx(R_j) \in \N$.
        
        \item
        \label{inv-one-app}
        Exactly one reaction is applicable in $\vx$ (unless $\vx(A_i) = 1$ for a halting state $i$, in which case no reactions are applicable).
        % \item 
        % \label{inv-correct-execution}
        % % If $\vx$ is a transition point of an oscillator species in configuration $i$, then the reaction applicable in $\vx$ corresponds to line $i$ of $\mathcal{R}$. 
        % If the reaction applicable in $\vx$ has $C_i$ as a product, then the reaction corresponds to line $i$ of $\mathcal{R}$. \todo{KC: I feel like we need a better notion of ``correctness'' corresponds feels very weird}
    \end{enumerate}
    
    \noindent We proceed by induction on the number of flux one-line segments connecting transition points $\vi$ and $\vc$. 
    (By Observation~\ref{obs:parallel-segments-concat} we may assume each segment is not a multiple of the previous.)
    
    For the base case,
    we show these invariants hold at $\vi$.
    Invariant \ref{inv-one-app} is established for all transition points below, including $\vi$.
    By construction, the only register species present in $\vi$ is $R_\mathrm{in}$, with concentration $n \in \N$, so invariant \ref{inv-int} is satisfied.
    
    Now, we show the inductive case that if the invariants hold at a transition point $\vx$,
    then we can execute the one applicable reaction
    (guaranteed to exist by invariant \ref{inv-one-app} unless we have halted) with flux 1,
    and that this will reach the next transition point $\vy$, such that the invariants still hold.
    
    First, we claim that at any transition point, $\vx$ with oscillator species $O_i$ having $\vx(O_i)=1$, at most one reaction is possible, exactly 1 if $i$ is a non-halting state, and 0 if $i$ is a halting state and $O_i = A_i$.
    If $O_i$ is $B_i$ or $C_i$, this is evident by the fact that each of those is a reactant in exactly one reaction in the network,
    and at transition points all other oscillator species are absent.
    In the case $O_i = A_i$,
    this is again evident if $i$ represents an increment instruction, since the 
    $C_i \rxni{B_i} A_{i+1} + R_j$
    reaction is the only one with $C_i$ as a reactant.
    If $i$ is a decrement, 
    then $C_i$ is a reactant in two reactions
    $C_i + R_j \rxni{B_i} A_{i+1}$
    and
    $C_i \rxni{B_i,R_j} A_k$,
    but one has $R_j$ as a reactant, and the other has $R_j$ as an inhibitor,
    so exactly one of those two reactions is applicable.
    This establishes that invariant \ref{inv-one-app} holds at the next transition point reached,
    when the applicable reaction is executed for one unit of flux.
    By Observation~\ref{obs:parallel-segments-concat} we assume a single segment applies this reaction until it is inapplicable, reaching the next transition point.

    It remains to argue that invariant \ref{inv-int} also holds at the next transition point.
    Let $O_i \in \{A_1,B_1,C_1,\ldots,A_m,B_m,C_m\}$ be an oscillator species and let $\vx = \{1O_i, m_1R_1,m_2R_2\ldots m_nR_n\}$, be a transition point that is reached from $\vi$.
    Assume the induction hypothesis that invariants \ref{inv-int} and \ref{inv-one-app} hold at $\vx$.
    Then each $m_i \in \N$ by invariant \ref{inv-int}.
    By \ref{inv-one-app} $\vx$ has exactly one applicable reaction.
    If $\vx$ is a transition point of $C_i$, and state $i$ of the register machine $\mathcal{R}$ is \lstinline{inc r_j}, then the applicable reaction in $\vx$ is 
    $
       \alpha = C_i \rxni{B_i} A_{i+1} + R_j.
    $
    % Such a reaction exists in $\mathcal{C}$ by construction. 
    By construction, there is no other reaction with $C_i$ as a reactant, and the reaction $A_{i+1} \rxni{C_{j_1},\dots,C_i, \ldots C_{j_l}} B_{i+1}$ (where each $C_{j_x}$ is a potential predecessor state of $i+1$) is inhibited by $C_i$, so every sufficiently long path from $\vx$ just executes $\alpha$ until we reach the transition point 
    $
    \vy = \{1A_{i+1}, m_1R_1,\ldots,(m_j+1)R_j, \ldots, m_nR_n \}.
    $
    So \ref{inv-int} holds. 
    
    If line $i$ of $\mathcal{R}$ is instead \lstinline{dec r_j,k} then there are reactions 
    \begin{align*}
        \beta_1 = & C_i + R_j \rxni{B_i} A_{i+1}
        \\
        \beta_2 = & C_i \rxni{B_i,R_j} A_{k}
        % \\
        % B_{i+1} &\rxni{A_i} C_{i+1}
    \end{align*}
    When $R_j$ is present in $\vx$, the only applicable reaction is $\beta_1$. 
    By a similar argument to the previous case, the reaction $A_{i+1} \rxni{C_i} B_{i+1}$ is inhibited, so every sufficiently long path from $\vi$ executes $\beta_1$ until we reach $\{1B_{i+1}$, $m_1R_1,$ $\ldots,$ $(m_j-1)R_j, \ldots, m_nR_n\}$.
    If $R_j$ is not present then the only applicable reaction is now $\beta_2$.
    Then every sufficiently long path from $\vi$ reaches the transition point $\{1B_k,m_1R_1,\ldots m_nR_n\}$.
    In either case, invariant \ref{inv-int} holds.
    This establishes the claim that invariants \ref{inv-int} and \ref{inv-one-app} hold at each reachable transition point.

    We now show that the sequence of states for oscillator species aligns with the execution order of lines in $\mathcal{R}$ and results in a correct simulation of $\mathcal{R}$.
    By construction of $\mathcal{C}$, for each line $i$ of $\mathcal{R}$ of the form \lstinline{inc r_j}, $\mathcal{C}$ has corresponding reaction $C_i \rxni{B_i} A_{i+1} + R_j$.
    By invariant \ref{inv-one-app}, this reaction will be applicable when transition point $\vc$ has a $C_i$ species present, so the next transition point will contain species $A_{i+1}$.
    Thus when $\mathcal{R}$ goes from line $i$ to $i+1$, the present oscillator species in $\mathcal{C}$ simulates this transition in the sense that the subscript $i$ is updated to $i+1$,
    and the concentration of $R_j$ increases by 1.
    Similarly, for each line $i$ of $\mathcal{R}$ of the form \lstinline{dec r_j,k}, 
    there are reactions 
    $C_i + R_j \rxni{B_i} A_{i+1}$
    and
    $C_i \rxni{B_i,R_j} A_{k}$.
    When $R_j$ is present, species $A_{i+1}$ is 1, and $R_j$ decreases by 1 at the next transition point, and when $R_j$ is not present, $A_k$ is 1.
    Thus decrement reactions are also properly simulated by $\mathcal{C}.$
    
    Since $\mathcal{R}$ halts with its output register having value $f(n)$, and $\mathcal{C}$ simulates $\mathcal{R}$, by \ref{inv-one-app} any sufficiently long sequence of reactions will eventually reach some static configuration $\vh$ representing $\mathcal{R}$'s halting configuration. 
    Furthermore, by \ref{inv-int} the values of the register species at the halting point are equal to the values of the registers in $\mathcal{R}$ when it halts. 
    Thus the configuration contains the correct concentration of $R_{\text{out}}$.
    Since this is a static (thus stable) configuration,
    this shows that $\mathcal{C}$ stably computes $f$.
\qed\end{proof}

    % As in the proof of Lemma~\ref{lem:oscillator},
    % we argue that at any time exactly one reaction is applicable,
    % unless $[A_i] = 1$ for a configuration $i$ representing a \lstinline{halt} instruction, in which case no reaction is applicable.
    % Thus any sufficiently long trajectory from that point will execute the reaction until its reactant is gone. 
    % Since exactly one reaction is applicable at any time,
    % at most two of any oscillator species is present at any time, 
    % and at the transition point when one reaction becomes inapplicable,
    % exactly one oscillator species is present (and has concentration 1). 

\subsection{Example of construction of iCRN from register machine}
\label{sec:example}
We demonstrate an example of our construction by translating a register machine $\mathcal{R}$ that computes the function $f(n) = 2n$ to an iCRN $\mathcal{C}$. 
% It is worth noting that $f$ is a non-linear function, so it is not computable by a non-inhibitory chemical reaction network as shown in ~\cite{chen2023rate}. 
The machine $\mathcal{R}$ that computes $f$ requires only input and output registers $r_{\text{in}},r_{\text{out}}$. 

{\centering
    \begin{tabular}{|l|l|}
        \hline
        {Instructions} & {Reactions} 
        \\ \hline
        1: 
        \verb!dec r_in,5!
        &  
        $\begin{array}{l}
        A_1 \rxni{C_4} B_1 \\
        B_1 \rxni{A_1} C_1 \\
        C_1 + R_\text{in} \rxni{B_1} A_{2} \\
        C_1 \rxni{B_1,R_\text{in}} A_5
        \end{array}$
        \\ \hline
        2: \verb!inc r_out!
        &
        $\begin{array}{l}
        A_2 \rxni{C_1} B_2 \\
        B_2 \rxni{A_2} C_2 \\
        C_2 \rxni{B_2} A_{3} + R_\text{out}
        \end{array}$
        \\ \hline
        3: \verb!inc r_out!
        &
        $\begin{array}{l}
        A_3 \rxni{C_2} B_3 \\
        B_3 \rxni{A_3} C_3 \\
        C_3 \rxni{B_3} A_{4} + R_\text{out}
        \end{array}$
        \\ \hline
        4: \verb!goto 1!
        &
        $\begin{array}{l}
        A_4 \rxni{C_3} B_4 \\
        B_4 \rxni{A_4} C_4 \\
        C_4 \rxni{B_4} A_1
        \end{array}$
        \\ \hline
        5: \verb!halt! & no reactions 
        \\ \hline
    \end{tabular}
\par}

\begin{figure}
    \centering
    \includegraphics[width=4.7in]{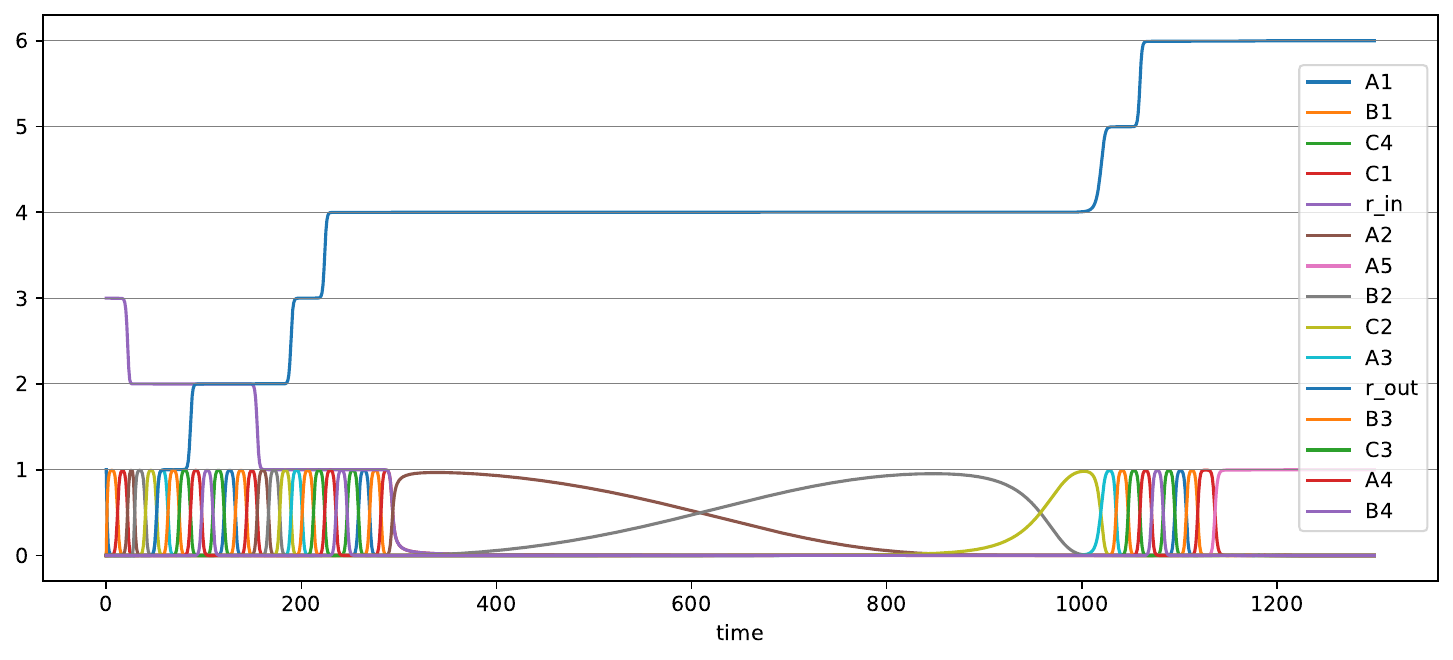}
    \caption{Plot of iCRN simulating ``multiply-by-2'' register machine, with input register \texttt{r\_in} having initial value $3$. 
    Note the species \texttt{r\_in} decrements from 3 down to 0, and the species \texttt{r\_out} increments from 0 up to 6, while other species oscillate.}
    \label{fig:mul2}
\end{figure}

\Cref{fig:mul2} shows a plot of this iCRN's trajectory,
under the mass-action rate model for reactants,
and where each inhibitor $I$ contributes a term $1 / (1 + 10^5 \cdot [I])$ to the rate of the reaction, as an approximation of ``absolute'' inhibition.\footnote{
    The long wave seen in the middle is because the reaction 
    $C_1 + R_\mathrm{in} \rxni{B_1} A_2$,
    when $R_\mathrm{in}$ starts at 1,
    has a much slower rate of convergence (linear, compared to exponential convergence when $R_\mathrm{in}$ starts 2 or higher).
    Consequently, $C_1$ from time $\approx 300$ to time $\approx 800$, despite being ``close'' to 0, is decaying to 0 much more slowly than in previous oscillations.
    Thus $C_1$ much more strongly inhibits the reaction $A_2 \rxni{C_1} B_2$ than in previous oscillations.
    $A_2$ and $B_2$ are the two species ``swapping'' very slowly between time 300 and 900.
}

% \begin{verbatim}
%     A1 + X --C1--> B2
%     A1  --C1,X--> BH
%     B2 --A2--> C1
%     C1 --B1--> A1

%     A2 --C2-> Y + B1
%     B2 --A1-->C2
%     C2 --B2-->A2
% \end{verbatim}

% \begin{verbatim}
% RD is output

% # dec RA and end if 0
% A1 + RA --C1--> B2
% A1 --C1,RA--> BH
% B1 --A_Cloop--> C1
% C1 --B1--> A1

% # dec RB and go to C loop if 0
% A2 + RB -->C2--> B3
% A2 --C2,RB--> B_Cloop
% B2 --A1,A4--> C2
% C2 --B2--> A2 

% # inc RC
% A3 --C3--> B4 + RC
% B3 --A2--> C3
% C3 --B3--> A3

% # inc RD and return to B loop
% A4 --C4--> B2 + RD
% B4 --A3--> C4
% C4 --B4--> A4

% #CloopTime

% #dec RC and go to A loop if 0
% A_Cloop + RC --C_Cloop--> B5
% A_Cloop --C_Cloop,RC--> A1
% B_Cloop --A2,A5--> C_Cloop
% C_Cloop --> B_Cloop --> A_Cloop

% #inc RB
% A5 --C3--> B_Cloop + RB
% B5 --A_Cloop--> C5
% C5 --B5--> A5

% #DD: when I entered it and printed them, it prints this:
% A1+RA --C1--> B2
% A1 --C1,RA--> BH
% B1 --AL--> C1
% C1 --B1--> A1
% A2+RB --C2--> B3
% A2 --C1,RB--> BL
% B2 --A1,A4--> C1
% C2 --B2--> A2
% A3 --C3--> B4+RC
% B3 --A2--> C3
% C3 --B3--> A3
% A4 --C4--> B2+RD
% B4 --A3--> C4
% C4 --B4--> A4
% AL+RC --CL--> B4
% AL --CL,RC--> A1
% BL --A2--> CL
% CL --BL--> AL
% A5 --C3--> BL+RB
% B5 --AL--> C5
% C5 --B5--> A5
% \end{verbatim}
\section{Conclusion}
\label{sec:conclusion}

There are some interesting questions for future research.

\paragraph{Relaxing absolute inhibition.}
The most glaring shortcoming of the inhibitory CRN model is the notion of ``absolute'' inhibition: 
any positive concentration of an inhibitor completely disables the reaction.
This is clearly unrealistic when taken to extremes:
with an enormous amount of reactant $R$, a tiny amount of $I$ cannot be expected to stop all $R$ from reacting via $R \rxn^I \dots$.
A more realistic model might say that the rate of a reaction is an increasing function of the concentration of its reactants and a decreasing function of the concentration of its inhibitors, for example using a Hill function such as $\frac{[R]}{1 + [I]}$ for the rate of the reaction.
However, any way of doing this seems to talk about rates,
and it is not clear how to meaningfully ask what tasks can be done in a rate-independent way in such a model.
One possible way to study this question meaningfully is similar to an approach suggested in the Conclusions of~\cite{chen2023rate} (for studying rate-independence in mass-action CRNs):
define a mass-action-like rate law in which a reaction's rate is a decreasing function of its inhibitors' concentrations,
and allow the adversary to set constant parameters in the rate law, but not to change the rate law itself.

\paragraph{Characterizing real-valued functions.}
We have demonstrated that the iCRN model is Turing universal in the sense that it can compute any computable function $f: \N \to \N$.
However, the natural data type for continuous iCRNs to process is real numbers.
It remains to characterize what functions $f: \Rp \to \Rp$ (or $f: \Rp^d \to \Rp$) can be stably computed by continuous iCRNs.
Using a \emph{dual-rail encoding} to encode a value $x$ as the difference of two concentrations $[X^+] - [X^-]$, one can also meaningfully investigate computation of functions $f:\R^d\to\R$ with negative inputs and outputs,
similar to the characterization of continuous piecewise linear functions stably computable by continuous (non-inhibitory) CRNs using dual-rail encoding~\cite{chen2023rate}.

\paragraph{Acknowledgements.}
DD and KC were supported by NSF awards 2211793, 1900931, 1844976, and DoE EXPRESS award SC0024467.

\bibliography{tam}

\begin{thebibliography}{1}
\providecommand{\url}[1]{\texttt{#1}}
\providecommand{\urlprefix}{URL }
\providecommand{\doi}[1]{https://doi.org/#1}

\bibitem{angluin2006fast}
Angluin, D., Aspnes, J., Eisenstat, D.: Fast computation by population protocols with a leader. Distributed Computing  \textbf{21}(3),  183--199 (Sep 2008), preliminary version appeared in DISC 2006

\bibitem{bournez2017odes}
Bournez, O., Gra\c{c}a, D.S., Pouly, A.: Polynomial time corresponds to solutions of polynomial ordinary differential equations of polynomial length. J. ACM  \textbf{64}(6) (oct 2017). \doi{10.1145/3127496}, \url{https://doi.org/10.1145/3127496}

\bibitem{chen2023rate}
Chen, H.L., Doty, D., Reeves, W., Soloveichik, D.: Rate-independent computation in continuous chemical reaction networks. Journal of the ACM  \textbf{70}(3) (May 2023). \doi{10.1145/3590776}

\bibitem{cooper2017constructing}
Cooper, C., Lamani, A., Viglietta, G., Yamashita, M., Yamauchi, Y.: Constructing self-stabilizing oscillators in population protocols. Information and Computation  \textbf{255},  336--351 (2017)

\bibitem{fages2017strong}
Fages, F., Le~Guludec, G., Bournez, O., Pouly, A.: Strong {T}uring completeness of continuous chemical reaction networks and compilation of mixed analog-digital programs. In: International conference on computational methods in systems biology. pp. 108--127. Springer (2017)

\bibitem{minsky1967computation}
Minsky, M.L.: Computation. Prentice-Hall Englewood Cliffs (1967)

\bibitem{SolCooWinBru08}
Soloveichik, D., Cook, M., Winfree, E., Bruck, J.: Computation with finite stochastic chemical reaction networks. Natural Computing  \textbf{7}(4),  615--633 (2008), \url{http://dx.doi.org/10.1007/s11047-008-9067-y}

\bibitem{srinivas2017enzyme}
Srinivas, N., Parkin, J., Seelig, G., Winfree, E., Soloveichik, D.: Enzyme-free nucleic acid dynamical systems. Science  \textbf{358}(6369) (2017)

\end{thebibliography}

% \newpage
% \appendix
% \input{appendix}

\end{document}